\documentclass[twocolumn,showpacs,preprintnumbers,amsmath,amssymb]{revtex4}

\usepackage{graphicx}
\usepackage{amsmath}
\usepackage{amssymb}
\usepackage{amsthm}

\newcommand{\be}{\begin{equation}}
\newcommand{\ee}{\end{equation}}
\newcommand{\ba}{\begin{array}}
\newcommand{\ea}{\end{array}}
\newcommand{\bea}{\begin{eqnarray}}
\newcommand{\eea}{\end{eqnarray}}

\newcommand{\calC}{{\cal C }}

\newcommand{\calL}{{\cal L }}

\newcommand{\calT}{{\cal T }}

\newcommand{\calS}{{\cal S }}
\newcommand{\calG}{{\cal G }}
\newcommand{\calP}{{\cal P }}
\newcommand{\calJ}{{\cal J }}
\newcommand{\calW}{{\cal W }}

\newcommand{\ZZ}{\mathbb{Z}}
\newcommand{\CC}{\mathbb{C}}

\newcommand{\la}{\langle}
\newcommand{\ra}{\rangle}

\newtheorem{dfn}{Definition}
\newtheorem{lemma}{Lemma}
\newtheorem{prop}{Proposition}
\newtheorem{theorem}{Theorem}

\begin{document}

\title{A no-go theorem for a two-dimensional self-correcting quantum memory based on stabilizer codes}

\author{Sergey Bravyi}\email{sbravyi@us.ibm.com}
\affiliation{\it IBM Watson Research Center}

\author{Barbara M. Terhal}\email{bterhal@gmail.com}
\affiliation{\it IBM Watson Research Center}

\begin{abstract}
We study properties of stabilizer codes that permit a local
description on a regular $D$-dimensional lattice.  Specifically, we
assume that the stabilizer group of a code (the gauge group for
subsystem codes) can be generated by local Pauli operators such that
the support of any generator is bounded by a hypercube of size
$O(1)$. Our first result concerns the optimal scaling of the
distance $d$ with the linear size of the lattice $L$. We prove an
upper bound $d=O(L^{D-1})$ which is tight for $D=1,2$. This bound
applies to both subspace and subsystem stabilizer codes. Secondly,
we analyze the suitability of stabilizer codes for building a
self-correcting quantum memory. Any stabilizer code with
geometrically local generators can be naturally transformed to a
local Hamiltonian penalizing states that violate the stabilizer
condition. A degenerate ground-state of this Hamiltonian corresponds
to the logical subspace of the code. We prove that for $D=1,2$
the height of the energy barrier separating different logical states is
upper bounded by a constant independent of the lattice size $L$.
The same result holds if there are 
unused logical qubits that are treated as ``gauge qubits". It
demonstrates that a self-correcting quantum memory cannot be built
using stabilizer codes in dimensions $D=1,2$. This result is in
sharp contrast with the existence of a classical self-correcting
memory in the form of a two-dimensional ferromagnet.
Our results leave open the possibility for a self-correcting quantum
memory based on 2D subsystem codes or on 3D subspace or subsystem
codes.
\end{abstract}

\maketitle

\section{Introduction}
One of the most intriguing open problems in the field of quantum
information processing is whether one can build a self-correcting
quantum memory --- a macroscopic physical system storing quantum
information for macroscopic periods of time without active error
correction \cite{Bacon:2005}. If it exists, such a hypothetical
device could play the role of a ``quantum hard drive" in a future
quantum computing technology. It is also a question of fundamental
interest how to design a system which allows for quantum coherence to
be preserved at non-zero temperature, similar in some sense to the
quest to understand the origin of high-temperature
superconductivity.


A possible scenario of how quantum self-correction can be achieved
in realistic physical systems was proposed by
Kitaev~\cite{Kitaev:1997}. The main idea of~\cite{Kitaev:1997} was
to mimic self-correction in classical magnetic media storage devices
using quantum spin or electron Hamiltonians with topologically
ordered ground-states. Logical qubits encoded into the ground-state
of such Hamiltonian become virtually isolated from the environment
assuming that the environment can only probe the system locally by
applying small static perturbations to the Hamiltonian. Note that
what is usually called a topological order translates to coding
language as simply the condition for a quantum code to have a
macroscopic distance.
A number of generalizations and proposals how this scenario could be
implemented in the lab have been made
recently~\cite{Freedman:2004,Kitaev:2005,jiang:2008,gladchenko:2008,Freedman:2008}.

As was first explicitly noted in \cite{Dennis:2001}, if the system
interacts with a thermal bath, the presence of a topological order
in the ground-state by itself does not guarantee self-correction.
Indeed, one may expect that local errors created by
thermal fluctuations will tend to accumulate which may result in a
global error having a non-trivial effect on the encoded qubits. This
phenomenon of `thermal fragility' of a topological order was 
studied by Nussinov and Ortiz in~\cite{Nussinov:2008}.
For Kitaev's 2D toric code the loss of topological order at any
non-zero temperature was studied in \cite{Nussinov:2008,castelnovo:2007, AFH:2008, kay:nonreliable}.

The main challenge in constructing a self-correcting memory is to
combine the existence of a coding ground-space with a natural,
passive, energy dissipation mechanism which prevents errors from
accumulating. It was argued by Dennis et al.~\cite{Dennis:2001} that
the 4D version of the toric code introduced by Kitaev
in~\cite{Kitaev:1997} is indeed a self-correcting quantum memory. In
\cite{Bacon:2005} Bacon considered the question of a self-correcting
memory in dimensions less than 4. He introduced a three-dimensional
model, based on a subsystem code which is now called the
three-dimensional Bacon-Shor code, which could be a model of
self-correcting quantum memory.


In order to set the stage for our results on a possible
self-correcting quantum memory, it is instructive to review the
classical state of affairs. A classical 1D Ising ferromagnet has two
degenerate ground-states corresponding to the classical repetition
code. The distance of this code is $n$, the number of spins. A
classical 1D ferromagnet is not a classical self-correcting memory,
since the phase transition to a disordered phase occurs at zero
temperature. The intuitive reason is that the energy cost of a
single domain of flipped spins is independent of the size of the
domain and hence such domain can grow without cost destroying the
order. A 2D classical Ising ferromagnet is a self-correcting memory;
the model exhibits a low-temperature memory phase separated by a
phase transition to a disordered phase \footnote{Note that the Ising
models are in fact not stable against perturbations (stray magnetic
fields etc.) while Kitaev's toric code model is stable against
sufficiently weak perturbations.}.

The 1D and 2D Ising ferromagnets allow for storage of only classical
information, since the distance of the classical repetition code
with respect to phase-flip errors is 1. Guided by this example, one
can take the following conditions as necessary for a {\it classical}
spin system to be a self-correcting memory:
\begin{enumerate}
\item The system has a degenerate ground-state; one has to flip a macroscopic
number of spins in order to map one ground-state to another.
\item A macroscopic energy barrier has to be traversed by any sequence of
single-spin flips mapping one ground-state to another.
\end{enumerate}
Note that if condition~1 is violated, the environment can destroy
the encoded information by acting only on a few spins. If
condition~2 is violated, there is no reason to expect that an energy
dissipation mechanism will prevent single-spin errors from
accumulating into a logical error. For example, the 1D Ising model
satisfies condition~1 but fails to satisfy condition~2. On the other
hand, the 2D Ising model satisfies both conditions since mapping one
ground-state (all spins up) to the other (all spins down) requires
creating a domain wall of macroscopic size.

Let us now consider a system made up from $n$ quantum spins (qubits)
with a degenerate ground-state spanning a linear subspace
$\calL\subseteq (\CC^2)^{\otimes n}$ which will be used to encode
logical qubits~\footnote{The assumption that logical qubits are
encoded into the ground-state of the system may fail for some models
such as topological quantum computation where information is encoded
into a collective state of anyonic excitations rather than the
ground-state.}.
The subspace $\calL$ can be viewed as the logical subspace of a
quantum code.

In the present paper we shall restrict ourselves to  {\it stabilizer
codes} (a.k.a. additive codes). The main idea of stabilizer codes is
to encode $k$ logical qubits into $n$ physical qubits using a
logical subspace $\calL\subseteq (\CC^2)^{\otimes n}$ spanned by
states $|\psi\ra$ that are invariant under the action of a {\it
stabilizer group} $\calS$,
\[
\calL=\{|\psi\ra\in (\CC^2)^{\otimes n}\, : \, P\, |\psi\ra=
|\psi\ra\quad \forall P\in \calS\}.
\]
Here $\calS$ may be an arbitrary Abelian subgroup of the Pauli group
\[
\calP=\la iI, X_1,Z_1,\ldots,X_n,Z_n\ra
\]
such that $-I\notin \calS$.
For any
stabilizer group $\calS$ one can always choose a set of generators
$\calS=\la S_1,\ldots,S_m\ra$ such that $S_a\in \calP$ are
self-adjoint Pauli operators. Thus the logical subspace $\calL$ can
be regarded as the ground-space of a Hamiltonian~\footnote{Our
analysis can be easily generalized to Hamiltonians $H=-\sum_{a=1}^m
r_a S_a$ with arbitrary real coefficients $r_a$.} \be \label{H}
H=-\sum_{a=1}^m S_a. \ee Note that all terms in $H$ pairwise
commute. In addition, the Hamiltonian has a constant gap above the
ground-space. All eigenvalues and eigenvectors of $H$ can be
explicitly computed using
the stabilizer formalism~\cite{Gottesman:PhD}. If all generators are independent, one has
$m=n-k$. In general $m>n-k$ since it may be advantageous to choose
an over-complete set of generators in order to make the energy
barrier higher.

We will assume that the physical qubits live at vertices of a
$D$-dimensional lattice $\Lambda=\{1,\ldots,L\}^{D}$ with open or
periodic boundary conditions. This choice for regular
$D$-dimensional lattice is made for pedagogical reasons, since all
our results can be generalized straightforwardly to an arbitrary
graph embedded into a $D$-dimensional space. Accordingly, there are
$n=L^D$ physical qubits. We would like the quantum memory
Hamiltonian defined in Eq.~(\ref{H}) to be physically realizable,
thus we demand that it involves only short-range interactions. More
precisely, we assume that the support of every generator $S_a$ can
be bounded by a hypercube with $r^D$ vertices for some constant
interaction range $r$.

Examples of such Hamiltonians are the family of surface codes \cite{Bravyi:1998}, color
 codes~\cite{bombin_delgado:2006,Bombin:2007'}, and 3D membrane-net models~\cite{Bombin:2007}.

Recall that the {\em weight} $|P|$ of a Pauli operator  $P=P_1
\ldots P_n \in \calP$ is the number of non-identity single-qubit
Pauli operators $P_i$.
The distance $d$ of a stabilizer code is the minimum weight of an
undetectable Pauli error, i.e., an operator $P\in \calP$ preserving
the logical subspace $\calL$ and inducing a non-trivial
transformation on $\calL$, 
\[
d=\min_{P\in \calC(\calS)\backslash \calS} |P|.
\]
Here $\calC(\calS)$ is the centralizer of the subgroup $\calS$,
i.e., a group of Pauli operators commuting with every element of $\calS$.
The notation ${\calC}(\calS)\backslash \calS$
means a subset of elements in ${\calC}(\calS)$ which are not in ${\calS}$.



Let the system described by Hamiltonian Eq. (\ref{H}) interact with
a thermal bath.
In order to decide whether the system is a good quantum memory, one
in principle would have to choose a dynamical model describing the
interaction with a thermal bath, then choose a particular decoding
algorithm and calculate how the probability of failure at the
decoding step scales with the system size. For realistic dynamical
models such as Davies' weak coupling
limit~\cite{Davies:1974,Alicki:2007}, pursuing this approach seems
to be completely hopeless.
However, it is not unreasonable to assume that the bath dephases the
system on short time-scales $T_2 \lll T_1$ \footnote{Such dephasing
in the energy eigenbasis can be explicitly enforced by measuring the
stabilizer generators $S_a$ using an ancilla qubit. We do not need
to learn the outcome of this measurement, since the role of this
measurement is only to project onto a certain eigenvalue of the
stabilizer}. Hence, after a short dephasing time, the state of the
system is a mixture of energy eigenstates and the effect of the bath
is a process of energy exchange with the system. The state of the
system can thus be represented as a probabilistic mixture of energy
eigenstates $E\, |\psi\ra$, where $|\psi\ra\in \calL$ is the encoded
logical state and $E\in \calP$  is a Pauli error.
An error $E\in \calP$ destroys encoded information if it is
undetectable by the code $\calL$, i.e., if
the restriction of $E$ onto $\calL$ induces a non-trivial transformation
on $\calL$.
Under these assumptions we can now state
 the quantum analogue of conditions~1,2 which are necessary for a
quantum spin system to be a self-correcting quantum memory:
\begin{enumerate}
\item[1q.] The ground-space of a (Hamiltonian) system is the
logical subspace of a quantum error correcting code with macroscopic
distance.
\item[2q.] A macroscopic energy barrier has to be traversed by any sequence of
single-qubit Pauli errors resulting in an undetectable error.
\end{enumerate}

The question that we address in this paper is: what quantum error
correcting codes are compatible with conditions 1q and 2q?
Note again that our conditions are necessary conditions for a
self-correcting quantum memory; we do not claim that they are also
sufficient conditions. In particular, whether a model is a
self-correcting quantum memory may depend on details of the scaling
of entropy versus energy at non-zero temperature.




\subsection{Bounds on The Distance}
Our first result concerns the optimal scaling of the distance $d$
with the linear size of the lattice $L$. It allows one to check
whether a code is compatible with condition~1q. In
Section~\ref{sec:stabilizer} we prove the following theorem

\begin{theorem}
\label{thm:1} Let $\calS=\la S_1,\ldots,S_m\ra$ be a stabilizer code
on a $D$-dimensional lattice
$\Lambda=\{1,\ldots,L\}^D$. Suppose the support of any generator $S_a$ can be
bounded by a hypercube with $r^D$ vertices.
 Then the distance of $\calS$ satisfies
\be \label{bound1} d\le r L^{D-1}. \ee
This bound holds for both periodic and open boundary conditions.
\end{theorem}

This theorem implies that any stabilizer code on a 1D lattice fails to satisfy
condition~1q. It contrasts the fact that for classical 1D codes the
distance can be of order $L$ (consider as an example the 1D
repetition code). For a 2D lattice the bound Eq.~(\ref{bound1})
allows $d$ to grow linearly with $L$ which is compatible with
condition~1q. Surface codes~\cite{Bravyi:1998} provide an example of
2D codes for which the bound Eq.~(\ref{bound1}) is saturated.

We are not aware of any codes saturating the bound
Eq.~(\ref{bound1}) for $D\ge 3$. Note that a generalized toric code
 on a $D$-dimensional lattice~\footnote{A generalized toric code on a $D$-dimensional
lattice has qubits living on hypercubes of dimension $D'\equiv
\lfloor D/2\rfloor$, $X$-type stabilizers living on hypercubes of
dimension $D'+1$, and $Z$-type stabilizers living on hypercubes of
dimension $D'-1$. Logical $X$-type operators correspond to closed
hypersurfaces of dimension $D'$ and logical $Z$-type operators
correspond to closed hypersurfaces of dimension $D-D'$ on the dual
lattice.} has distance $d\sim L^{\lfloor D/2\rfloor}$.
We conjecture that in $D=3$ the distance of any stabilizer
code with local generators satisfies $d=O(L)$.

If a stabilizer code has more than one logical qubit, some of the logical qubits may be
protected from the environment better than the others. Note
that Theorem~\ref{thm:1} bounds the minimum weight of logical operators for the {\em worst}
choice of a logical qubit. One can ask whether the same bound applies to the {\em best} logical qubit
as well. To state this question more formally  assume that we encode $k$ qubits
using a stabilizer code $\calS$ with $g+k$ logical qubits for some $g>0$. We
will regard the extra $g$ logical qubits with the corresponding
logical Pauli operators
$\overline{X}_1,\overline{Z}_1,\ldots,\overline{X}_g,\overline{Z}_g$
as unused ``gauge qubits"~\cite{Poulin:2005} such that any error
affecting only the gauge qubits can be ignored. 
At the same time we can use the logical Pauli operators on the gauge qubits to minimize
the weight of `useful' logical operators.
Thus, in this case the relevant distance would be the minimum weight
of a Pauli operator $P$ preserving the logical subspace, i.e., $P\in
\calC(\calS)$, that cannot be generated by stabilizers and logical
Pauli operators on
the gauge qubits, that is,
 \be \label{dgauge} d({\calG})=\min_{P\in
\calC(\calS)\backslash \calG} \; |P|, \quad \calG=\la
\calS,\overline{X}_1,\overline{Z}_1,\ldots,\overline{X}_g,\overline{Z}_g\ra.
\ee
 In Section~\ref{subs:subsystem_aux} we will use the formalism of
subsystem codes to prove the following extension of Theorem~\ref{thm:1}.

\vspace{3mm}

\noindent {\bf Theorem~1${}^*$.} {\it Under the assumptions of
Theorem~\ref{thm:1} the distance $d(\calG)$ satisfies the bound
Eq.~(\ref{bound1}) for any choice of logical operators
$\overline{X}_1,\overline{Z}_1,\ldots,\overline{X}_g,\overline{Z}_g$
on the gauge qubits.}

\vspace{3mm}

Let us remark that an analogous bound does not apply to `bare' logical operators, i.e.,
$P\in \calC(\calG)\backslash \calG$. The minimal weight of `bare' logical operators
can be of order $L^D$, see \footnote{If we would only
consider the minimum weight of elements of ${\calC}(\cal
G)\backslash {\cal G}$, that is the weight of `bare' logical
operators on the logical qubits, one can easily construct 1D
counterexamples to Theorem~1$^*$. Take (odd) $n$ copies of a small
code, such as the Steane code [[7,1,3]] and define two high-weight
logical operators as ${\bf X}=\overline{X}_1 \overline{X}_2 \ldots
\overline{X}_n$ and ${\bf Z}=\overline{Z}_1 \overline{Z}_2 \ldots
\overline{Z}_n$. Here $\overline{X}_i,\overline{Z}_i$ are the
logical operators for the $i$th Steane code. It is clear that even
when we multiply these operators ${\bf X}$ and ${\bf Z}$ with
elements in $\calS$, their minimum weight will scale with $n$. Aside
from this logical qubit there are also gauge qubits. If we multiply
${\bf X},{\bf Z}$ with the logical operators of these gauge qubits,
we can reduce the weight of ${\bf X}$ or ${\bf Z}$ to $O(1)$.}.  

\subsection{Bounds on The Energy Barrier}
\label{subs:into_bounds}
Our second result concerns the scaling of the energy barrier with
lattice size $L$. In contrast to the distance $d$ the energy barrier
is not a property of the stabilizer group $\calS$ only --- it
depends on the choice of generators $S_a$ used to define the
Hamiltonian Eq.~(\ref{H}). As was mentioned above, the set of
generators can be vastly overcomplete (consider as an example the 4D
toric code). In order to exclude overcomplete generating sets in
which some generator $S_a$ appears a macroscopic number of times we
shall impose a (natural) restriction that any qubit can participate
only in a constant number of generators $S_a$.

Note that the Hamiltonian Eq.~(\ref{H}) has a ground-state energy
$-m$. Let $|\psi\ra\in \calL$ be any ground-state and $E\in \calP$
be any Pauli operator. Since $E$ either commutes or anti-commutes
with every term in the Hamiltonian, a state $E\, |\psi\ra$ has
energy $\la \psi|E^\dag H E|\psi\ra =-m + \epsilon(E)$, where
$\epsilon(E)$ is proportional to the number of
Pauli operators $S_a$ in the Hamiltonian $H$
which anticommute with $E$,
\be
\epsilon(E)=2\, \#\{ a\, : \, S_a E=-E S_a\}.
\label{eq:ecost}
\ee
 We shall refer to $\epsilon(E)$ as the {\it
energy cost} of a Pauli operator $E$.

Let us say that a sequence $P_0,P_1,\ldots,P_t\in \calP$ is a walk
on the Pauli group starting at $P_0$ and arriving at $P_t$ iff for
all $i$ the operators $P_i$ and $P_{i+1}$ differ by at most one
qubit. Let $\calW(S,T)$ be a set of all walks starting at $S$ and
arriving at $T$. For any walk $\gamma\in \calW(S,T)$ let
$\epsilon_{max}(\gamma)$ be the maximum energy reached by $\gamma$
\[
\epsilon_{max}(\gamma)=\max_{P\in \gamma} \epsilon(P).
\]
Suppose the environment tries to implement a Pauli error $E\in
\calP$ by a sequence of single-qubit Pauli errors. The minimum
amount of energy the environment has to inject into the system in
order to implement $E$ is
\[
d^\ddag(E)=\min_{\gamma\in \calW(I,E)} \epsilon_{max}(\gamma).
\]
Thus the energy barrier the environment has to overcome in order to
implement a non-trivial logical operator on the encoded qubits is
 \be
\label{d*} d^\ddag = \min_{E\in \calC(\calS)\backslash \calS} \; \;
\min_{\gamma\in \calW(I,E)} \epsilon_{max}(\gamma). \ee Let us
emphasize once more that in contrast to the distance $d$ the energy
barrier $d^\ddag$ is not uniquely determined by the stabilizer group $\calS$ but
is a function of the generating set $S_1,\ldots,S_m$ used to construct the Hamiltonian
Eq.~(\ref{H}).


Note that our assumption that the environment implements a logical
operator by {\em single-qubit} Pauli errors is not a restriction if
we are interested in determining the scaling of the energy barrier
$d^{\ddag}$ with $L$. To see this, imagine that the environment
instead makes Pauli errors with weight at most $w=O(1)$ at the time.
One can always simulate a single step of such a walk by $w$ steps
with single-qubit Pauli errors. It can increase the maximum energy
reached by any walk by at most a constant and hence
the energy barrier $d^\ddag$ can change at most by a constant.

Using the assumption that any qubit participates in $O(1)$
generators $S_a$ one can easily prove a naive upper bound \be
\label{naive} d^\ddag=O(d). \ee Therefore, if we want $d^\ddag$ to
grow with $L$ we must look for codes with the distance $d$ growing
with $L$. For 2D codes one may have $d\sim L$ and thus the naive
upper bound Eq.~(\ref{naive}) cannot rule out a possibility that
$d^\ddag$ grows with $L$.
In Section~\ref{sec:stabilizer} we prove
the following theorem:
\begin{theorem}
\label{thm:2} Let $\calS=\la S_1,\ldots,S_m\ra$ be a stabilizer code
with local generators on a 2D lattice such that each qubit
participates in a constant number of generators.  Then the energy
barrier $d^\ddag$ is upper bounded by a constant independent of the
lattice size $L$.
\end{theorem}
This theorem tells us that 2D stabilizer codes cannot be compatible
with condition~2q and thus the corresponding Hamiltonians defined in
Eq.~(\ref{H}) cannot be a self-correcting quantum memory. This
result contrasts the fact that 2D classical codes with local
generators (e.g. the 2D Ising model) can be used to build a
self-correcting classical memory. We remark that Theorem~\ref{thm:2}
applies to both open and periodic boundary conditions. It can also
be easily generalized to quasi-2D lattices with the number of 2D
layers bounded by a constant independent of $L$.

Theorem~\ref{thm:2} can be strengthened for stabilizer codes which
encode more than one logical qubit. Specifically, we can encode $k$
qubits using a stabilizer code with $g+k$ logical qubits by
regarding the extra $g$ logical qubits  as ``gauge qubits", see the
discussion after Theorem~\ref{thm:1}. Accordingly, any error
affecting only the gauge qubits can be ignored. In this case the
energy barrier that the environment has to overcome in order to
implement a non-trivial logical operator would be \be
\label{dstargauge} d^\ddag(\calG)= \min_{E\in \calC(\calS)\backslash
\calG} \; \; \min_{\gamma\in \calW(I,E)} \epsilon_{max}(\gamma), \ee
where $\calG=\la \calS,
\overline{X}_1,\overline{Z}_1,\ldots,\overline{X}_g,\overline{Z}_g\ra$
is a group generated by $\calS$ and the logical Pauli operators on
the gauge qubits. In Section~\ref{subs:subsystem_aux} we will use
the formalism of subsystem codes to prove the following extension of
Theorem~\ref{thm:2}.

\vspace{3mm}

\noindent {\bf Theorem~2${}^*$.} {\it Under the assumptions of
Theorem~\ref{thm:2} the energy barrier $d^\ddag(\calG)$ is upper
bounded by a constant independent of the lattice size $L$ for any
choice of the gauge Pauli operators
$\overline{X}_1,\overline{Z}_1,\ldots,\overline{X}_g,\overline{Z}_g$.}

\vspace{3mm}

It has been observed by many authors~\cite{Bacon:2005,Alicki:2007}
that the 2D toric code does not feature a macroscopic energy barrier
since the logical operators for this code have a string-like
geometry. The energy cost of a partially implemented logical
operator comes only from the two end-points of a string which cannot
be larger than some small constant. One can ask whether logical
operators for general 2D stabilizer codes can always be chosen as
``strings" (in which case the logical operators can be interpreted
as moving point-like anyons around the lattice).
The proof of
Theorem~\ref{thm:2} presented in  Section~\ref{sec:stabilizer} provides
a partial answer to this question. It implies that for any 2D
stabilizer code with generators of size $r\times r$ there exists at least one logical operator
whose support can be covered by  a
rectangle of size $r\times L$. Such a rectangle can be regarded as a quasi-1D string.

\subsection{Beyond Subspace Stabilizer Codes}

Hamiltonians of the form in Eq. (\ref{H}) may seem like overly
restrictive models to consider for a quantum memory, in particular
since they involve a set of {\em commuting} operators. Common models
in many-body physics involve geometrically-local sets of
non-commuting operators, such as the Heisenberg model and its
variants, Kitaev's honeycomb model \cite{Kitaev:2005}, the quantum
compass model, see e.g. \cite{qcompass}, etc.
Such Hamiltonians can always be written as \be \label{Hgauge}
H=\sum_{a=1}^m r_a\, G_a, \ee where $r_a$ are real coefficients and
$G_a$ are elements of the Pauli group $\calP$ with local support.
The formalism of {\em stabilizer subsystem codes}~\cite{Poulin:2005}
provides a systematic way of constructing Hamiltonians of the form
Eq.~(\ref{Hgauge}) with a degenerate ground-state that can be taken
as possible models of a quantum memory.
This formalism focuses on those
symmetries of the Hamiltonian that can be described using the Pauli group ignoring
all other non-Pauli symmetries such as the $SU(2)$-symmetry (since they cannot be 
analyzed using the framework of stabilizer codes).

By analogy with stabilizer codes which can be fully characterized by
the stabilizer group, a subsystem code can be characterized by its
{\it gauge group}. In our case the gauge group $\calG$ is the group
generated by the Pauli operators $G_a$, see Eq.~(\ref{Hgauge}).
Conversely, any (non-abelian) subgroup $\calG\subseteq \calP$ with
local generators $G_1,\ldots,G_m$ yields a family of Hamiltonians
Eq.~(\ref{Hgauge}).
The symmetries of the Hamiltonian $H$ defined in Eq.~(\ref{Hgauge})
are described by the centralizer of $\calG$, i.e. those elements in
$\calP$ which commute with every element in $\calG$.  It
can be always represented as
\[
\calC(\calG)=\la
\calS,\overline{X}_j,\overline{Z}_j,j=1,\ldots,k\ra,
\]
where $\calS$ is the {\em center} of the gauge group, i.e.
$\calS= \calG \cap \calC(\calG)$, and
$\overline{X}_j,\overline{Z}_j$ are {\it logical Pauli operators} of
the quantum code. Since the symmetries of $H$ include the logical
Pauli group on the $k$ logical qubits, any eigenvalue of $H$ has
degeneracy at least $2^k$. In particular, the ground-state of $H$
can be used to encode $k$ logical qubits.

The group $\calS$ is called a {\it stabilizer group} of the
subsystem code. It induces a decomposition of the Hilbert space of
$n$ qubits into $\calG$-invariant sectors, $(\CC^2)^{\otimes
n}=\bigoplus_{\bf s} \calL_{\bf s}$, where ${\bf s}$ labels
different syndromes, i.e., irreps of the stabilizer group. Given any
syndrome ${\bf s}$, the corresponding sector $\calL\equiv \calL_{\bf
s}$ defines a {\it logical subspace} of the subsystem
code~\footnote{It is well known that the codes corresponding to
different choices of the syndrome ${\bf s}$ are equivalent up to a
local change of basis in every qubit.}. The logical subspace
possesses a subsystem structure, $\calL=\calL_{logical}\otimes
\calL_{gauge}$, where $\calL_{logical}$ defines the {\it logical
subsystem} and $\calL_{gauge}$ defines the {\it gauge subsystem}.
The gauge group $\calG$ restricted to $\calL$ can be identified with
the group of Pauli operators on the gauge subsystem $\calL_{gauge}$.
The group generated by the logical operators
$\overline{X}_j,\overline{Z}_j$ restricted to $\calL$ can be
identified with the group of Pauli operators on the logical
subsystem $\calL_{logical}$. Accordingly,
$\dim{(\calL_{logical})}=2^k$.

It is important to note that the stabilizer group of a subsystem
code with local generators does not necessarily have local
generators. This fact may make such subsystem codes potentially more
powerful than stabilizer codes. Subsystem codes with an Abelian
gauge group have no gauge qubits, i.e., $\calL=\calL_{logical}$
hence they coincide with the (subspace) stabilizer codes discussed
above. Note that for (subspace) stabilizer codes the stabilizer
group $\calS$ coincides with the gauge group $\calG$.

The distance $d$ of a subsystem code is defined as the minimum
weight of a Pauli error $P\in \calP$ preserving $\calL$
and inducing a non-trivial transformation on $\calL_{logical}$,
 \be
 \label{subsystemd} d=
 \min_{P\in \calC(\calS)\backslash \calG}  \,
|P|, \quad \calS=\calG\cap \calC(\calG).
 \ee
Let us mention a useful identity
\be
\label{useful_identity}
\calC(\calS)=\calG\cdot \calC(\calG)
\ee
that holds for any subsystem code $\calG$. 
Note that a subsystem code with a large distance must have a
non-trivial stabilizer group. Indeed, if $\calS=\la I\ra$, one has
two possibilities: (i) any single-qubit Pauli operator belongs to
$\calG$, and (ii) some single-qubit Pauli operator is a non-trivial
logical operator. In case~(i) one has $\calG=\calP$, i.e., the code
has no logical qubits. In case~(ii) one has $d=1$.

The best known example of a locally generated subsystem code is the
Bacon-Shor code~\cite{Bacon:2005, AC:bs} which can be defined on 2D
or 3D lattice.
 The Hamiltonian $H$ associated with the 2D Bacon-Shor code corresponds
to the quantum compass model in condensed-matter physics,
\[
H_{BS}=-\sum_{1\le i,j\le L} J_x X_{i,j} X_{i+1,j} + J_z  Z_{i,j} Z_{i,j+1}.
\]
Here each pair $i,j$ represents a vertex of the lattice and
$J_x,J_z$ are arbitrary real coefficients. Accordingly, the gauge
group $\calG$ of the 2D Bacon-Shor code has generators  $X_{i,j}
X_{i+1,j}$ and $Z_{i,j} Z_{i,j+1}$. This code has distance $L$. The
stabilizer group $\calS$ is generated by operators $S^x_i
=\prod_{j=1}^L X_{i,j} X_{i+1,j}$ and $S^z_j=\prod_{i=1}^L Z_{i,j}
Z_{i,j+1}$. One can easily check that all elements of $\calS$ have
weight at least $2L$, that is, $\calS$ cannot have local generators.
The code has one logical qubit with logical Pauli operators
$\overline{X}=\prod_{j=1}^L X_{1,j}$ and $\overline{Z}=\prod_{i=1}^L
Z_{i,1}$.

The example of the Bacon-Shor code demonstrates that some subsystem codes with local
generators originate from subspace codes with highly non-local stabilizer groups.
It might suggest that subsystem codes can beat the upper bound Eq.~(\ref{bound1}).
We will show that this intuition is wrong by proving
\begin{theorem}
\label{thm:3}
Let $\calG=\la G_1,\ldots,G_m\ra$ be the gauge group of
a subsystem stabilizer code on a $D$-dimensional lattice
$\Lambda=\{1,\ldots,L\}^D$. Suppose the support of any
generator $G_a$ can be bounded by a hypercube with $r^D$ vertices.
Then the distance of $\calG$ satisfies 
 \be \label{bound1sub} d\le 3r L^{D-1}. \ee
\end{theorem}
Note that the distance of a subsystem code depends only on the gauge group $\calG$.
Therefore, condition~1q is true or false for any Hamiltonian in the family Eq.~(\ref{Hgauge})
regardless of the choice of coefficients $r_a$.

As far as condition~2q is concerned, the first natural question is
how to define the energy barrier $d^\ddag$ that must be traversed in
order to implement a logical operator.  In
Section~\ref{subs:subsystem_energy_barrier} we propose a definition
of $d^\ddag$  that depends on the energy spectrum of the Hamiltonian
in different sectors $\calL_{\bf s}$. We also derive a usable upper
bound on the energy barrier that does not require computing the
energy spectrum. For any particular Hamiltonian this upper bound can
be calculated using only the stabilizer formalism.  Unfortunately,
we have not yet been able to prove an analogue of
Theorem~\ref{thm:2} for subsystem codes. In
Section~\ref{subs:subsystem_energy_barrier} we prove a very weak
upper bound on $d^\ddag$, namely, $d^\ddag=O(d)$. This leaves as an
open question whether Hamiltonian models based on 2D subsystem codes
may have a macroscopic energy barrier.

\subsection{Organization of The Paper}

The proof of Theorems~\ref{thm:1}-\ref{thm:3} relies on three technical results, Lemma's \ref{lemma:1}-\ref{lemma:3}.
The first Lemma \ref{lemma:1} which we call a ``cleaning lemma"
asserts that for any logical operator $P\in \calC(\calS)$ and for
any subset of qubits $M$, $|M|<d$, one can choose
a stabilizer $S\in \calS$ such that $PS$  acts trivially on $M$.
The stabilizer $S$ involves only those generators $S_a$ 
whose support overlaps with $M$.
Lemma \ref{lemma:2} proves an analogue of this cleaning lemma for
subsystem codes. The third Lemma \ref{lemma:3} which we call a
``restriction lemma" relates the distance of a code $\calS$ defined
on the entire lattice $\Lambda$ and the distance of a code $\calS_M$
obtained by restricting generators of $\calS$ onto some subset of
qubits $M\subseteq \Lambda$. The restriction lemma applies to both
subspace and subsystem stabilizer codes.

Stabilizer codes with local generators are discussed in
Section~\ref{sec:stabilizer} which proves Theorems~\ref{thm:1} and
\ref{thm:2}. Subsystem codes with local generators are discussed in
Section~\ref{sec:subsystem}. We prove the upper bound on the
distance of subsystem codes (Theorem~\ref{thm:3}) in
Section~\ref{subs:subsystem_distance}. Upper bounds on the energy
barrier of subsystem codes are discussed in
Section~\ref{subs:subsystem_energy_barrier}. Finally,
Section~\ref{subs:subsystem_aux}  shows how to prove
Theorems~1${}^*$ and 2${}^*$. Some open problems are discussed in
Section~\ref{sec:open}.


\section{Stabilizer Codes}
\label{sec:stabilizer}

Recall that $\calP$ denotes the Pauli group on $n$ qubits.
We assume that the qubits live at vertices of a $D$-dimensional lattice
$\Lambda=\{1,\ldots,L\}^D$.
For any subgroup $\calS\subseteq \calP$ and any subset $M\subseteq \Lambda$
introduce a group
\[
\calS(M)=\{ P\in \calS\, : \, \mbox{Supp}(P)\subseteq M\}
\]
which includes all elements of $\calS$ whose support is contained in
$M$. In particular, $\calP(M)$ is a group of all Pauli operators
whose support is contained in $M$. Introduce also a group
\[
\calS_M=\{ P\in \calP(M)\, : \, PQ\in \calS \quad \mbox{for
some $Q\in \calP(\Lambda \backslash M)$}\}
\]
which includes all Pauli operators $P\in \calP(M)$  that can be
extended to some element of $\calS$.
In other words  $\calS_M$ is a group obtained by 
restricting elements in $\calS$ to $M$. 
Note that if some element $P\in \calS$ crosses the boundary of $M$ then
the restriction of $P$ onto $M$ is no longer element of $\calS$.
By definition $\calS(M)\subseteq \calS_M\subseteq
\calP(M)\subseteq \calP$.



Our main technical tool will be the following ``cleaning" lemma. It
allows one to clean out any region $M\subset \Lambda$ of size
smaller than the distance such that no logical operator of the code
contains $X$, $Y$, or $Z$ on qubits of $M$. More formally, one can
multiply any logical operator $P\in \calC(\calS)$ by a stabilizer
$S\in \calS$  such that $PS$ acts trivially on $M$. The stabilizer
$S$ uses only those generators $S_a$ whose support overlaps with
$M$, see Figure~\ref{fig:cleaning1}. The cleaning lemma is
particularly useful when the generators of $\calS$ are local. In
this case the cleaning changes $P$ only inside $M$ and in a small
neighborhood of the boundary of $M$. Thus the cleaning of $P$ can be
done multiple times, so that multiple `holes' can be made into the
support of $P$.
\begin{figure}
\centerline{\includegraphics[height=2.8cm]{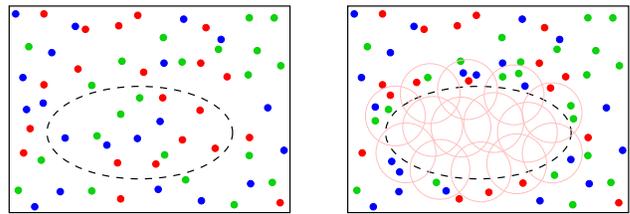}}
\caption{Support of a logical operator $P\in \calC(\calS)$ is shown
by the colored dots (representing $X$, $Y$, and $Z$ operators). In
order to clean out the region $M$ (interior of the dashed ellipse)
one can multiply $P$ with a stabilizer  $S\in \calS$. The stabilizer
$S$ includes only those generators $S_a$  whose support overlaps
with $M$. It yields an equivalent logical operator $PS$ acting
trivially on $M$. The cleaning can performed simultaneously for
different logical operators and for different spacial regions. }
\label{fig:cleaning1}
\end{figure}

\begin{lemma}[Cleaning Lemma]
\label{lemma:1}
Let $\calS=\la S_1,\ldots,S_m\ra$ be a stabilizer code and
$M\subseteq \Lambda$ be an arbitrary subset of qubits.
Denote $\calJ(M)$ a set of indexes $a$ such that the support of $S_a$
overlaps with $M$. Then one of the following is true:\\
(1) There exists a non-trivial logical operator $P\in \calC(\calS)\backslash \calS$
whose support is contained in $M$,\\
(2) For any logical operator $P\in \calC(\calS)$ one can choose a stabilizer
\[
S=\prod_{a\in \calJ(M)} S_a^{x_a}, \quad x_a\in \{0,1\}
\]
such that $PS$ acts trivially on qubits of $M$.
\end{lemma}

\begin{proof}
Let $\la iI\ra\subseteq \calP$ be a subgroup of phase factors.
By definition of the subgroups $\calS_M$ and $\calS(M)$ one has the following inclusion:
\be
\label{inclusion1}
\la iI\ra \cdot \calS(M) \subseteq \calC(\calS_M)\cap \calP(M).
\ee
If the inclusion Eq.~(\ref{inclusion1}) is strict, there exists a Pauli operator $P\in \calP(M)$
such that $P\in \calC(\calS_M)$ but $P\notin \calS$, that is, $P\in \calC(\calS)\backslash \calS$.
It corresponds to the case~(1) in the statement of the lemma.
Otherwise, Eq.~(\ref{inclusion1})  is an equality, that is,
$\calC(\calS_M)\cap \calP(M)=\la iI\ra \cdot \calS(M)$.
 Taking the centralizer of both sides
one gets
\be
\label{inclusion2}
\calS_M=\calC(\calS(M))\cap \calP(M).
\ee
(Note that for any subgroup $\calT$ of the Pauli group one has 
$\calC(\calC(\calT))=\la iI\ra \cdot \calT$.)
Let $P\in \calC(\calS)$  be any logical operator and
$P_M$ be a restriction of $P$ onto $M$.
By definition, $P_M\in \calC(\calS(M))\cap \calP(M)$.
It follows from Eq.~(\ref{inclusion2}) that $P_M\in \calS_M$, that is,
one can find a stabilizer $S\in \calS$ such that $PS$ acts trivially on $M$.
Let us expand the stabilizer $S$ in term of the generators and take out all
generators whose support does not overlap with $M$. It yields a new stabilizer
$S'\in \calS$ such that $PS'$ still acts trivially on $M$.
On the other hand, $S'$ is a product of generators $S_a$, $a\in \calJ(M)$, so
we arrive at the case~(2) in the statement of the lemma.
\end{proof}

The proof of both Theorems~\ref{thm:1} and \ref{thm:2} goes by
applying the cleaning lemma to get an upper bound on a {\it linear
distance} of a code that we define below. Let us first define a
linear distance of a code when the lattice
$\Lambda=\{1,\ldots,L\}^D$ has open boundary conditions.
\begin{dfn}
\label{dfn:ldistance}
Given a Pauli operator $P\in \calP$ define its {\it linear dimension} $d_1(P)$
as the minimum length of a contiguous interval $R\subseteq \{1,\ldots,L\}$
such that $R\times \{1,\ldots,L\}^{D-1}$ covers the support of $P$.
Given a stabilizer code $\calS$ define a linear distance of $\calS$ as
\be
\label{d_1}
d_1(\calS)=\min_{P\in \calC(\calS)\backslash \calS} d_1(P).
\ee
\end{dfn}
Thus a linear distance of a code characterizes the minimal linear
dimension of non-trivial logical operators along some fixed
coordinate axis. Above we have chosen the first coordinate axis but
one could similarly define a linear distance for any other
coordinate axis. For periodic boundary conditions we define a linear
dimension $d_1(P)$ as the minimum length of a contiguous  interval
$R\subseteq \ZZ_L$ such that $R\times (\ZZ_L)^{D-1}$ covers the
support of $P$. Then a linear distance of a code is defined by the
same formula Eq.~(\ref{d_1}).

\begin{prop}
\label{prop:ldistance}
Let $\calS=\la S_1,\ldots,S_m\ra$ be a stabilizer code on a lattice
$\Lambda=\{1,\ldots,L\}^D$
 with open or periodic boundary conditions.
 Suppose the support of any generator $S_a$ can be covered by a hypercube with $r^D$ vertices. Then
\be
\label{ldistance_upper}
d_1(\calS)\le r
\ee
for any $L\ge 2(r-1)^2$.
\end{prop}
\begin{proof}
We shall prove the proposition for $D=2$. The generalization to
other dimensions is straightforward. Let us assume that the bound
Eq.~(\ref{ldistance_upper}) is not satisfied, that is, \be
\label{contrary} d_1(\calS)>r \ee and show that it leads to a
contradiction.The idea behind the proof is depicted in
Fig.~\ref{fig:cleaning2}.

One can easily check that any integer $L\ge 2(r-1)^2$ can be represented as
$L=a(r-1)+br$ for some integers $a,b\ge 0$ such that $K\equiv a+b$ is even.
Therefore we can  represent the lattice $\Lambda$
as a disjoint union of vertical strips,  $\Lambda=A_1\cup \ldots \cup A_K$,
such that any strip has width (along the $x$-axis) $r$ or $r-1$,
and the total number of strips $K$ is even, see Fig.~\ref{fig:cleaning2}.
The assumption Eq.~(\ref{contrary}) implies that
any non-trivial logical operator $P \in \calC(\calS)\backslash \calS$
is supported on at least two strips.
Therefore,  Lemma~\ref{lemma:1} implies that we
can multiply $P$ with some product of generators $S_a$
to clean out any chosen even strip $A_2,A_4,\ldots,A_{K}$.
 However, a support of any generator $S_a$ has width at most $r$ and thus it
overlaps with at most one even strip, so we can apply  Lemma~\ref{lemma:1} repeatedly
to clean out all even strips. In other words,
 we can define a new logical operator
\[
P'=P \, S, \quad \mbox{for some} \quad S=\prod_{p=2,4,\ldots,K} \; \prod_{a\in \calJ(A_{p})} S_a^{x_a}
\]
such that $P'$ is supported only on odd strips, that is
\[
P'=P_1' P_3' \cdots  P_{K-1}', \quad P_j'\in \calP(A_j).
\]
Next we note that each generator $S_a$ overlaps with at most one odd
strip. Therefore the inclusion $P'\in \calC(\calS)$ is possible only
if $P_j'\in \calC(\calS)$ for all $j=1,3,\ldots,K-1$. Since
$P'\notin \calS$ there must exist at least one odd strip $A_{j}$
such that $P_j'\notin \calS$, that is, $P_j'$ is a non-trivial
logical operator. However it contradicts Eq.~(\ref{contrary}) since
the linear dimension $d_1(P_j)$ is at most $r$.
\end{proof}

\begin{figure}
\centerline{\includegraphics[height=3.4cm]{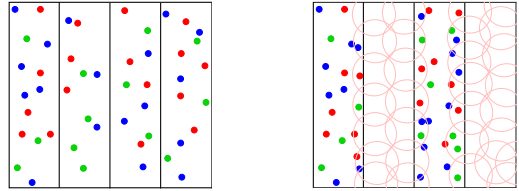}}
\caption{{\it Theorem~\ref{thm:1}: sketch of the proof.} Choose any
logical operator $P\in \calC(\calS)\backslash \calS$. Assuming
$d_1(\calS) > r$, we can clean out any vertical strip of width $r$.
Repeating this procedure on parallel disjoint strips (even strips),
we obtain an equivalent logical operator $P'$ which has support on
the remaining set of disjoint strips (odd strips). If the strips are
wide enough, then $P'$ commutes with $\calS$ on every odd strip
individually. It must be that the restriction of $P'$ on some odd
strip is not in $\calS$, otherwise $P'\in \calS$. Thus we obtain a
non-trivial logical operator whose support is contained in a single
vertical strip which contradicts the assumption that
$d_1(\calS)>r$.} \label{fig:cleaning2}
\end{figure}

Theorems~\ref{thm:1} and \ref{thm:2} are direct corollaries of
Proposition~\ref{prop:ldistance}. Indeed, choose any logical
operator $P\in \calC(\calS)\backslash \calS$ with a linear dimension
$d_1(P)\le r$. Then the weight of $P$ is at most $rL^{D-1}$, that
is, $d\le rL^{D-1}$ which proves Theorem~\ref{thm:1}. Furthermore,
if $D=2$ then the bound $d_1(P)\le r$ implies that the support of
$P$ can be covered by a quasi-1D vertical string. One can construct
a walk $\gamma\in \calW(I,P)$ that implements $P$ in a row-by-row
fashion. At any intermediate step the contribution to the energy
cost of a partially implemented $P$ comes only from the two
end-points of the string. The assumption that any qubit participates
in $O(1)$ generators $S_a$ implies that
$\epsilon_{max}(\gamma)=O(1)$. It proves Theorem~\ref{thm:2}.

{\it Remark:} After having tried to construct 3D stabilizer codes
with $d=O(L^2)$ which would saturate the bound in
Eq.~(\ref{bound1}), we conjecture that any 3D stabilizer code with
local generators satisfies $d=O(L)$. In fact, all known examples of
stabilizer codes correspond to a bound $d=O(L^{\lfloor D/2
\rfloor})$. The pattern that seems to emerge is that there is a
trade-off between the minimum weight of, say, the logical
$\overline{X}$ versus logical $\overline{Z}$ operator. The 1D
repetition code shows that one can have a logical $\overline{X}$ of
weight $O(L)$, but then the logical $\overline{Z}$ is of weight
$O(1)$,  so the distance of the code is $O(1)$. In 2D, one can make
both operators of weight $O(L)$ as the surface codes demonstrate. In
3D, the 3D toric  code is an example where again one logical
operator is of weight $O(L^2)$ but the other logical operator is of
weight $O(L)$. In 4D, both operators are of weight $O(L^2)$.


\section{Subsystem Codes}
\label{sec:subsystem}

\subsection{Bounds on The Code Distance}
\label{subs:subsystem_distance} Let us start by generalizing the
cleaning lemma to subsystem codes.
\begin{lemma}[Cleaning Lemma for Subsystem Codes]
\label{lemma:2} Let $\calG$ be a gauge group of a subsystem code
and $M$ be an arbitrary subset of qubits.
Then one of the following is true:\\
(1) There exists a non-trivial logical operator $P\in \calC(\calS)\backslash \calG$
whose support is contained in $M$,\\
(2) For any logical operator $P\in \calC(\calG)$ one can choose
a stabilizer $S\in \calS$ such that $PS$ acts trivially on  $M$.
\end{lemma}
\begin{proof}
Let $\la iI\ra\subseteq \calP$ be a subgroup of phase factors.
By definition of the subgroups $\calS_M$ and $\calG(M)$ one has the
following inclusion:
\be
\label{inclusion3}
\la iI\ra \cdot \calG(M)\subseteq
 \calC(\calS_M)\cap \calP(M).
\ee If the inclusion Eq.~(\ref{inclusion3}) is strict, there exists
a Pauli operator $P\in \calP(M)$ such that $P\in \calC(\calS_M)$ but
$P\notin \calG$, that is, $P\in \calC(\calS)\backslash \calG$. It
corresponds to case~(1) in the statement of the lemma. Otherwise
Eq.~(\ref{inclusion3}) is an equality, that is, $\la iI\ra \cdot
\calG(M)= \calC(\calS_M)\cap \calP(M)$. Taking the centralizer of
both sides one gets \be \label{inclusion4}
 \calS_M=\calC(\calG(M))\cap \calP(M).
\ee
Let $P\in \calC(\calG)$  be any logical operator and
$P_M$ be a restriction of $P$ onto $M$.
By definition, $P_M\in \calC(\calG(M))\cap \calP(M)$. It follows from Eq.~(\ref{inclusion4})
that $P_M\in \calS_M$, that is,  one can find a stabilizer $S\in \calS$ such that $PS$ acts trivially on $M$.
\end{proof}
Note that cleaning out the region $M$ may require changing the operator $P$ globally
(not only in a small neighborhood of $M$) since the stabilizer group $\calS$ of a
subsystem code may lack a local generating set.  This is the reason why for subsystem codes
we cannot clean out multiple disconnected regions simultaneously: cleaning out one region by multiplying
with elements in $\calS$ may generate support in another `far-away' region which we already cleaned out.

On the other hand, if the stabilizer group $\calS=\calG \cap
\calC(\calG)$ has local generators, $\calS=\la S_1,\ldots,S_m\ra$,
the operator $S$ in Lemma~\ref{lemma:2} can be chosen such that it
includes only those generators $S_a$ whose support overlaps with
$M$. In this case the cleaning can be performed ``locally".

For subsystem codes we shall need a new lemma which relates the
distance of a code $\calG$ defined on the entire lattice $\Lambda$
to the distance of a code $\calG_M$ obtained by restricting $\calG$
onto some subset of qubits $M\subseteq \Lambda$
(here we use notations introduced in Section~\ref{sec:stabilizer}).
Recall that $l_\infty$-distance between vertices $u,v\in \Lambda$
is defined as $\max_{j=1}^D |u_j-v_j|$.
\begin{lemma}[Restriction Lemma]
\label{lemma:3} Let $\calG=\la G_1,\ldots,G_m\ra$ be a gauge group
of a subsystem code on a $D$-dimensional lattice $\Lambda$. 
Suppose the support of any
generator $G_a$ can be bounded by a hypercube with $r^D$ vertices.
Let $M\subseteq \Lambda$ be an arbitrary subset and $\partial M$ be
a subset of vertices $u\in \Lambda\backslash M$ such that
$l_\infty$-distance between $u$ and $M$ is at most $r$. Consider the
subsystem code with gauge group $\calG_M$.
Then one of the following is true:\\
(1) The code $\calG_M$ has no logical qubits,\\
(2) The code $\calG_M$ has distance at least $d-|\partial M|$\\
where $d$ is the distance of $\calG$ defined in Eq.~(\ref{subsystemd}).
\end{lemma}
{\it Remark:} By abuse of notation we regard the code $\calG_M$ as a
code defined on qubits of $M$ only (otherwise any Pauli operator on
a qubit $u\in \Lambda\backslash M$ would be a logical operator for
the code $\calG_M$). One can always extend the code $\calG_M$ to the
rest of the lattice in a trivial way, for example, by adding a gauge
operator $Z_u$ for any qubit $u\in \Lambda\backslash M$. For
simplicity we shall ignore these technicalities.
\begin{proof}
Indeed, suppose $\calG_M$ has at least one logical qubit and let
$d'$ be the distance of $\calG_M$. Then there exists a logical
operator $P\in \calC(\calG_M)\cap \calP(M)$, $P\notin \calG_M$ and a
gauge operator $G\in \calG_M$ such that $|PG|=d'$.
  Then $P\in \calC(\calG)\backslash \calG$, that is, $P$ is a non-trivial logical
operator for the original code. By definition of the group
$\calG_M$, we can extend the operator $G\in \calG_M$ beyond $M$ 
to some gauge operator $G' \in \calG$ such that $G$ is the restriction of $G'$
onto $M$. Since $\calG$ has local generators, such an operator $G'$ 
can be chosen to have support only in $M$ and in $\partial M$.
It means  that $|PG'|\leq
d'+|\partial M|$.
Since $PG'$ is a non-trivial logical operator for the code $\calG$,
its weight must be
at least $d$, or $d \leq d'+|\partial M|$.
\end{proof}

Now we are ready to apply these Lemma's to prove
Theorem~\ref{thm:3}. For simplicity we present a proof for $D=1$ and
open boundary conditions. Generalization to higher dimensions and
periodic boundary conditions is straightforward. In 1D we have to
prove the upper bound \be \label{2r} d\le 3r. \ee Let $M\subseteq
\Lambda$ be the smallest contiguous block of qubits such that the
subsystem code $\calG_M$ obtained by restricting $\calG$ onto $M$
has at least one logical qubit (if there are several minimal blocks
$M$, choose anyone of them). Let $d'$ be the distance of the code
$\calG_M$. The restriction lemma implies that \be \label{d-2r} d'\ge
d-2r. \ee Let us assume that $d'>r$ and show that it leads to a
contradiction. Indeed, $d'>r$ implies $|M|>r$, so we can partition
$M$ into three contiguous blocks, $M=ABC$, such that $|B|=r-1$ and
$A,C\ne \emptyset$. Let $P,Q\in \calC(\calG_M)\backslash \calG_M$ be
any pair of anti-commuting logical operators for the code $\calG_M$.
Since the size of $B$ is smaller than the distance of $\calG_M$, we
can apply Lemma~\ref{lemma:2} to the code $\calG_M$ to clean out
the region $B$ which yields a pair of ``cleaned" logical operators
\[
P,Q\in \calC(\calG_M)\backslash \calG_M, \quad PQ=-QP
\]
 such that
$P,Q$ act trivially on $B$, that is, $P=P_A P_C$, $Q=Q_A Q_C$.
Obviously, $PQ=-QP$ iff $P_A Q_A=-Q_A  P_A$ or $P_C Q_C =-Q_C P_C$.
Assume without loss of generality that $P_A$ anti-commutes with $Q_A$
(otherwise, relabel the blocks $A$ and $C$).
Since the code $\calG_M$ has local generators of size at most $r$ and
the block $B$ has length $r-1$, no generator of $\calG_M$ overlaps with both $A$ and $C$.
Therefore $P,Q\in \calC(\calG_M)$ implies $P_A,Q_A\in \calC(\calG_M)$.
Combining this
observation with the fact that $P_A$ and $Q_A$ anti-commute we conclude that
neither of
 $P_A,Q_A$ can be an element
  of $\calG_M$, that is
 \[
P_A,Q_A\in \calC(\calG_M)\backslash \calG_M, \quad P_AQ_A=-Q_AP_A.
\]
Consider the subsystem code $\calG_A$ obtained by restricting
$\calG_M$ onto $A$. We want to argue that $P_A,Q_A$ are non-trivial
logical operators for the code $\calG_A$. Indeed, by definition,
$P_A,Q_A\in \calC(\calG_M)$ implies $P_A,Q_A\in \calC(\calG_A)$.
Since $P_A$ and $Q_A$  anti-commute, neither of them can be an
element of $\calG_A$, that is,
\[
P_A,Q_A\in \calC(\calG_A)\backslash \calG_A.
\]
Summarizing we have found a contiguous block $A\subset M$,
$|A|<|M|$, such that the code $\calG_A$ has at least one logical
qubit. This however contradicts the minimality of $M$ and thus
proves that $d'\le r$. Taking into account Eq.~(\ref{d-2r}) we
arrive at $d\le 3r$.

The proof of Theorem \ref{thm:3} for 2D subsystem codes is analogous. In that case we would consider
the minimal width strips $M$ of height $L$ such that $\calG_M$ has at least one logical qubit.
The same approach generalizes to arbitrary dimensions $D$.

\subsection{Energy Barrier Upper Bounds}
\label{subs:subsystem_energy_barrier}
Let us start by describing some general properties of the
Hamiltonian $H$ defined in Eq.~(\ref{Hgauge}).
Recall that we define
a gauge group $\calG\subset \calP$ as a group generated by Pauli
operators $G_a$ in the Hamiltonian Eq.~(\ref{Hgauge}). We define a
stabilizer group $\calS$ as the center of $\calG$, i.e.,
$\calS=\calG\cap \calC(\calG)$. The stabilizer group induces a
decomposition of the Hilbert space of $n$ qubits into
$\calG$-invariant sectors, $(\CC^2)^{\otimes n}=\bigoplus_{\bf s}
\calL_{\bf s}$, where ${\bf s}$ labels different syndromes (irreps
of $\calS$). By definition, $H$ is block-diagonal with respect to
this decomposition,
\[
H=\bigoplus_{\bf s} H_{\bf s}.
\]
Each sector $\calL_{\bf s}$ possesses a subsystem structure
$\calL_{\bf s}=\calL_{logical,\bf s}\otimes \calL_{gauge, \bf s}$.
The Hamiltonian $H_{\bf s}$ acts trivially on the $2^k$-dimensional
logical subsystem $\calL_{logical,\bf s}$ and has a non-trivial
spectrum of gauge qubit excitations (which may be different for
different sectors $\bf s$).

The question of whether such Hamiltonian can be a model of a
self-correcting memory is more involved than for stabilizer codes
and goes beyond this paper. One of the first questions is to
determine whether the ground-space of the Hamiltonian is confined to
one sector, and if so, whether the excitations from this sector,
call it ${\bf s}_0$, to the other sectors are {\em gapped}. The
fulfillment of such conditions would show that the Hamiltonian has
some stability with respect to perturbations that affect the logical
qubits. This condition is naturally obeyed for stabilizer code
Hamiltonians. Note that we do not require the gauge-qubit
excitations in a fixed sector to be gapped, since these excitations
do not harm the state of the logical qubits.
In particular, there may be low-lying or gapless gauge-qubits
excitations from the ground-space which give rise to a local order
parameter, see for example the discussion of the quantum compass
model in \cite{Nussinov:2008}. The presence of a local order for the
gauge qubits, i.e. a local operator acting on the gauge qubits which
distinguishes two different ground-sectors, does not at all prevent
the possibility of a `topological order' for the logical qubits. It
only expressed the fact that the gauge qubits are not a good place
to put the quantum information.

Let us define the gap \be
 \label{gapdef}
  \Delta_{\bf s}=\lambda(H_{\bf s})-\lambda(H_{\bf s_0}),
  \ee
where $\lambda(H_{\bf s})$ is the smallest eigenvalue of $H_{\bf s}$
and ${\bf s}_0$ is the sector that contains the smallest eigenvalue of  $H$.

A proper definition of the energy barrier that must be traversed in
order to perform a logical operation is dependent on the
energy spectrum of the Hamiltonian in different ${\bf s}$ sectors. Consider a
logical operator
 $E \in \calC(\calS)\backslash \calG$ and let $\gamma \in
\calW(I,E)$ a walk on the Pauli group implementing $E$. In order to
define the energy cost of a Pauli operator $E$, we consider what the
operator does on the ${\bf s_0}$ sector. Any state in the ${\bf
s_0}$ sector will be mapped onto a state in some other fixed sector
${\bf s}$. We define the energy cost of $E$ as \be
\epsilon(E)=\Delta_{{\bf s}|E({\bf s_0})={\bf s}}.
\label{eq:defcostsub} \ee
Let the maximum energy barrier along a
path be $\epsilon_{max}(\gamma)=\max_{P \in \gamma}\epsilon(P)$.
Then we define the energy barrier for $E$ as
\[
d^\ddag(E)=\min_{\gamma\in \calW(I,E)} \epsilon_{max}(\gamma).
\]
The minimum over these energy barriers for any logical operation is then given by
 \be
\label{d*sub}
d^\ddag = \min_{E\in \calC(\calS)\backslash \calG} \; \;  d^{\ddag}(E).
\ee

In the definition of $d$ and $d^{\ddag}$ we minimize over elements
of $\calC(\calS)\backslash \calG$ which are comprised of a `bare' logical
operator (an element of $\calC(\calG)\backslash \calG$) {\em times} any gauge operator
(an element of $\calG$), see Eq.~(\ref{useful_identity}).
 The minimum weight of operators in ${\calC}(\calS)\backslash \calG$ may
be much less than the minimum weight of `bare' logical operators. A good example is the Heisenberg model on a
$D$-dimensional lattice  where ${\calG}$
is generated by nearest-neighbor $XX$, $YY$ and $ZZ$ operators. If
$L$ is odd, the operators $X_{\rm
all}=\prod_{u\in \Lambda} X_u$ and $Z_{\rm all}=\prod_{u\in \Lambda}
Z_u$ commute with ${\calG}$ but are not generated by ${\calG}$. In
addition we have ${\cal S}=I$. This implies that the minimal weight
of the bare logical operators in ${\calC}(\calG) \backslash \calG$
is $L^D$. However, the distance $d$ of this code
is $1$, because we can reduce the weight of, say, $X_{\rm all}$ to $1$
by multiplying with it with $XX$-type gauge operators. If some gauge qubits
operations are costly to implement (i.e. require creating
high-energy excitations) then it may be more physically reasonable
to not count these operations for free in minimizing the distance
and the energy barrier. Whether such choices are warranted however
depends on the features and spectrum of the Hamiltonian and cannot
be analyzed using the coding framework alone.


In order to set the stage for proving no-go results which do not depend on detailed
properties of the Hamiltonian, we will derive a simple upper bound
on the energy cost of a Pauli operator $E$ defined in
Eq.~(\ref{eq:defcostsub}). We apply this analysis to the Hamiltonian
defined in Eq.~(\ref{Hgauge}) assuming for simplicity that $r_a=-1$,
that is,
\be
 \label{Hgauge1}
 H=-\sum_{a=1}^m G_a.
\ee
Let $E\in \calP$ be any Pauli error mapping the ground state sector $\calL_{\bf s_0}$
to some sector $\calL_{\bf s}$ and let $|\psi_0\ra$ be a ground state
of $H_{\bf s_0}$, such that $\lambda(H_{\bf s_0})=\la \psi_0|H|\psi_0\ra$.
Since $E\, |\psi_0\ra\in \calL_{\bf s}$ we get an upper bound
 $\lambda(H_{\bf s})\le \la \psi_0|E^\dag H E |\psi_0\ra$. It implies
 \[
 \epsilon(E)=\lambda(H_{\bf s})-\lambda(H_{\bf s_0}) \le \la \psi_0|E^\dag H E - H|\psi_0\ra
 \]
and thus
\[
\epsilon(E)\le \| E^\dag H E-H\| \le \|\,  [H,E] \, \|.
\]
We can bound the norm of the commutator $[H,E]$ taking into account
that $E$ either commutes or anti-commutes with any generator $G_a$.
It yields \be \label{ubound} \epsilon(E)\le 2 \, \#\{ a\, : \, EG_a
=-G_a E\}. \ee Summarizing we can bound the energy cost of $E$ by
(twice) the number of terms in the Hamiltonian anti-commuting with
$E$. Note that this upper bound is identical to the definition of
energy cost for stabilizer codes, Eq.~(\ref{eq:ecost}).

The upper bound Eq.~(\ref{ubound}) opens up the possibility for a generalization of the
energy barrier arguments in Section \ref{sec:stabilizer}
to subsystem codes.
For example, one can use Eq.~(\ref{ubound}) to show that
$d^{\ddag}=O(1)$ for the 2D Bacon-Shor code. This easily follows
from the fact that a partially implemented logical operator (say,
$\overline{X}=\prod_{i} X_{1,i}$) only anti-commutes with the gauge
operators $ZZ$ at the end-points. (Numerical results
of~\cite{qcompass} indicate that $d^\ddag$ decays exponentially with
the lattice size $L$.)

Using Eq.~(\ref{ubound}) we can also show that $d^{\ddag}=O(d)$.
Indeed, choose a non-trivial logical operator $P=P_{u_1}\ldots
P_{u_d} \in \calC(\calS) \backslash \calG$ with weight $d$. The
number of generators $G_a$ anti-commuting with any single-qubit
operator $P_{u_i}$ is at most $O(1)$. Therefore the number of
generators anti-commuting with $P$ is at most $O(d)$.
Implementing  the sequence of the single-qubit Pauli operators
$P_{u_1},\ldots, P_{u_d}$ in an arbitrary order one gets a walk $\gamma\in \calW(I,P)$
with $\epsilon_{max}(\gamma)=O(1)$.

\subsection{Proofs of Theorem~1${}^*$ and 2${}^*$}
\label{subs:subsystem_aux}

Let $\calS=\la S_1,\ldots,S_m\ra$ be a stabilizer code with local
generators, and let
$\overline{X}_1,\overline{Z}_1,\ldots,\overline{X}_g,\overline{Z}_g$
be logical Pauli operators on some subset of $g$ logical qubits (the
ones we want to treat as gauge qubits). Let $\calG=\la
\calS,\overline{X}_1,\overline{Z}_1,\ldots,\overline{X}_g,\overline{Z}_g\ra$.
 We can regard $\calG$ as the gauge group of a subsystem
code. Clearly, $\calG\cap \calC(\calG)=\calS$, so $\calS$ is the
stabilizer group of $\calG$. In order to prove Theorems~1${}^*$ and
2${}^*$ let us define a linear distance of a subsystem code.
\begin{dfn}
\label{dfn:ldistance1}
Given a subsystem code $\calG$ define a linear distance of $\calG$ as
\be
d_1(\calG)=\min_{P\in \calC(\calS)\backslash \calG} d_1(P).
\ee
\end{dfn}
Recall that $d_1(P)$ is a linear dimension of $P$, that is, the minimal length of a contiguous
interval $R\subseteq \{1,\ldots,L\}$ such that $R\times \{1,\ldots,L\}^{D-1}$ covers the support of $P$,
see Section~\ref{sec:stabilizer}.
Thus a linear distance of a code characterizes the minimal linear dimension of non-trivial logical operators
along some fixed coordinate axis.

\begin{prop}
\label{prop:ldistance1}
Let $\calG$ be a subsystem code on a lattice $\Lambda=\{1,\ldots,L\}^D$
with open or periodic boundary conditions.
Assume that a stabilizer group $\calS=\calG\cap \calC(\calG)$ has local generators,
$\calS=\la S_1,\ldots,S_m\ra$, such  that
the support of any generator $S_a$
can be covered by a hypercube with $r^D$ vertices. Then
\be
\label{ldistance_upper1}
d_1(\calG)\le r
\ee
for any $L\ge 2(r-1)^2$.
\end{prop}
\begin{proof}
The proof is almost identical to  the proof of
Proposition~\ref{prop:ldistance} in Section~\ref{sec:stabilizer}.
It goes by assuming that $d_1(\calG)>r$ and showing that it leads to
a contradiction.
The only difference is that instead of cleaning out a logical
operator $P\in \calC(\calS)\backslash \calS$ using
Lemma~\ref{lemma:1} one has to clean out a logical operator $P\in
\calC(\calG)\backslash \calG$ using Lemma~\ref{lemma:2}. After
cleaning out all even strips $A_2,\ldots,A_K$ one gets an equivalent
logical operator $P'\in \calC(\calG)\backslash \calG$ that has
support only on odd strips $A_1,A_3,\ldots,A_{K-1}$ (we used the
notations from the proof of Proposition~\ref{prop:ldistance}). Since
any generator of $\calS$ overlaps with at most one odd strip, the
inclusion $P'\in \calC(\calG)\subseteq \calC(\calS)$ implies that a
restriction of $P'$ onto any odd strip is an element of
$\calC(\calS)$. At least one of these restrictions is not an element
of $\calG$. Therefore, a restriction of $P'$ onto some odd strip is
an element of $\calC(\calS)\backslash \calG$. 
This is in
contradiction with the assumption that $d_1(\calG)>r$.
\end{proof}

Theorems~1${}^*$ and 2${}^*$ are direct corollaries of
Proposition~\ref{prop:ldistance1}. Indeed, choose any logical
operator $P\in \calC(\calS)\backslash \calG$ with a linear dimension
$d_1(P)\le r$. Then the weight of $P$ is at most $rL^{D-1}$, that
is, $d(\calG)\le rL^{D-1}$ which proves Theorem~1${}^*$.
Furthermore, if $D=2$ then the bound $d_1(P)\le r$ implies that the
support of $P$ can be covered by a quasi-1D vertical string and thus
one can construct a walk $\gamma\in \calW(I,P)$ with
$\epsilon_{max}(\gamma)=O(1)$. It proves Theorem~2${}^*$.

\section{Discussion and open problems}
\label{sec:open} This paper addressed the problem of constructing a
self-correcting quantum memory based on stabilizer codes with
geometrically-local generators. We developed several technical tools
for analyzing such codes and proved upper bounds on the distance
that are tight for 1D and 2D codes. In addition, we defined an
energy barrier separating different logical states and proved that
for 2D stabilizer codes the energy barrier can not grow with the
lattice size.
We note that a similar conclusion has been independently reached by Kay and Colbeck~\cite{KC:selfqc}
using completely different techniques.

 It would be interesting to prove our conjecture that $d=O(L)$ for
stabilizer codes in 3D. In addition, we would like to bound the
energy barrier $d^{\ddag}$ for 3D stabilizer codes and for 2D
subsystem codes. It might be possible that a self-correcting quantum
memory can only be based on 3D subsystem codes such as the one
suggested by Bacon~\cite{Bacon:2005}. This would require a detailed
analysis of the energy barrier for these systems.

Another interesting problem concerns algorithms for computing the
distance of stabilizer codes with local generators. The dynamic
programming algorithm allows one to compute the distance of
stabilizer (subsystem) codes with local generators in time of order
$L\exp{(L^{D-1})}$. It requires exponential time for $D\ge 2$. It
remains an open question whether the distance can be computed in
polynomial time for $D\ge 2$.

\section{Acknowledgements}
We thank Dave Bacon for posing some of the questions that lead to
the results in this paper. We thank Panos Aliferis for many
discussions on the topic of this paper. We are grateful to 
Roger Colbeck for explaining to us several technical points of~\cite{KC:selfqc}.

%


\end{document}